\documentclass[12pt]{amsart}
\usepackage{amssymb}
\usepackage{color}
\usepackage{amsmath,epic,curves}
\usepackage[english]{babel}
\usepackage{graphicx}
\newtheorem{claim}{}[section]
\newtheorem{theorem}[claim]{Theorem}

\newtheorem{lemma}[claim]{Lemma}
\newtheorem{proposition}[claim]{Proposition}

\renewenvironment{proof}{\noindent{\it Proof. \hskip0pt}}
                      {$\square$\par\medskip}

\begin{document}
\baselineskip 5.6 truemm
\parindent 1.5 true pc

\newcommand\lan{\langle}
\newcommand\ran{\rangle}
\newcommand\tr{{\text{\rm Tr}}\,}
\newcommand\ot{\otimes}
\newcommand\ol{\overline}
\newcommand\join{\vee}
\newcommand\meet{\wedge}
\renewcommand\ker{{\text{\rm Ker}}\,}
\newcommand\image{{\text{\rm Im}}\,}
\newcommand\id{{\text{\rm id}}}
\newcommand\tp{{\text{\rm tp}}}
\newcommand\pr{\prime}
\newcommand\e{\epsilon}
\newcommand\la{\lambda}
\newcommand\inte{{\text{\rm int}}\,}
\newcommand\ttt{{\text{\rm t}}}
\newcommand\spa{{\text{\rm span}}\,}
\newcommand\conv{{\text{\rm conv}}\,}
\newcommand\rank{\ {\text{\rm rank of}}\ }
\newcommand\re{{\text{\rm Re}}\,}
\newcommand\ppt{\mathbb T}
\newcommand\rk{{\text{\rm rank}}\,}
\newcommand\bcolor{\color{blue}}
\newcommand\ecolor{\color{black}}
\newcommand\sss{\omega}

\title{Faces for two qubit separable states and the convex hulls of trigonometric moment curves}

\author{Seung-Hyeok Kye}
\address{Department of Mathematics and Institute of Mathematics\\Seoul National University\\Seoul 151-742, Korea}

\date\today

\thanks{partially supported by NRFK 2012-0000939}



\begin{abstract}
We analyze the facial structures of the convex set consisting of all two
qubit separable states. One of faces is a four dimensional convex
body generated by the trigonometric moment curve arising from
polyhedral combinatorics. Another one is an eight dimensional convex
body, which is the convex hull of a homeomorphic image of the two
dimensional sphere. Extreme points consist of points on the surface,
and any two of them determines an edge. We also reconstruct the trigonometric moment curve
in any even dimensional affine space using the qubit-qudit systems, and characterize
the facial structures of the convex hull.
\end{abstract}

\maketitle

\section{Introduction}\label{sec:intro}

Let $M_n$ denote the $C^*$-algebra of all $n\times n$ matrices over the complex field.
A state on $M_n$ is a unital positive linear functional on $M_n$, and is represented by a density matrix which is
a positive semi-definite matrix with trace one. A state
on the tensor product $M_m\ot M_n=M_{mn}$ is said to be separable if it is the convex sum of
rank one projections onto simple tensors of the form $x\ot y\in\mathbb C^m\ot\mathbb C^n$, which are called product vectors.
A state which is not separable is said to be entangled. The notion of entanglement arising from quantum mechanics turns out to be very useful
in the current quantum information and quantum communication theory.

The convex structures for the convex set $\mathbb S_{m\times n}$ of
all separable states on $M_m\ot M_n$ are highly nontrivial, and have
begun to be studied very recently. See
\cite{alfsen,alfsen_2,chen_dj_semialg,choi_kye,ha+kye_unique_decom,ha+kye_unique_choi_length}, for examples. For the
simplest case of $m=n=2$, all faces of the convex set $\mathbb
S_{2\times 2}$ have been classified in \cite{ha_kye_04}. The first purpose
of this note is to analyze convex sets arising as faces of $\mathbb
S_{2\times 2}$, and to report relations with polyhedral combinatorics.
The convex set $\mathbb S_{2\times 2}$
itself is a $15$-dimensional convex body whose extreme points are
parameterized by the $4$-dimensional manifold $\mathbb C\mathbb
P^1\times\mathbb C\mathbb P^1$. Two extreme points represented by
$(x,y)$ and $(z,w)$ on the manifold make an edge if and only if
$x\neq y$ and $z\neq w$. Otherwise, they determine a face of
$\mathbb S_{2\times 2}$ which is affinely isomorphic to the
$3$-dimensional solid sphere whose extreme points are parameterized
by just $\mathbb C\mathbb P^1$.

There are two kinds of maximal faces. One kind of them consists of $8$-dimensional convex sets whose extreme points are parameterized by
two dimensional sphere $\mathbb C\mathbb P^1$. The convex combination of any two extreme points is an edge.
If we take three extreme points, then the corresponding points on the sphere determine a circle.
It turns out that this circle determines a maximal face of the maximal face, and is of $4$-dimensional.
In this way, we have an embedding of the circle into the $4$-dimensional convex body $C^4$ with the property:
A point of $C^4$ is an extreme point if and only if it is on the image of the circle, and any two extreme points make an edge.
This embedding is exactly the trigonometric moment curve studied in polyhedral combinatorics.

The moment curve had been studied in \cite{gale} to realize four dimensional polytopes with the property that
convex combinations of any two extreme points are edges. It was shown that any finite distinct points on the moment curve
$$
(t,t^2,r^3,\dots,t^{2p})\in\mathbb R^{2p},\qquad t\in\mathbb R
$$
give rise to a convex body, called a cyclic polytope, with the
property that any choice of $p$ extreme points makes a face. The
trigonometric moment curve
\begin{equation}\label{tri-moment}
(\cos t,\sin t, \cos 2t,\sin 2t,\dots, \cos pt,\sin pt)\in\mathbb R^{2p},\qquad t\in\mathbb [0,2\pi)
\end{equation}
has the same property. The convex hull, denoted by $C^{2p}$, of the whole trigonometric moment curve has been studied in \cite{smil}
for the case of $p=2$. Extreme points of $C^4$ are those on the curve, and the convex hull of any two extreme points
is an edge. There are no more nontrivial faces of $C^4$. The convex hull of the moment curve has been studied in \cite{puente}.


The second purpose of this note is to investigate the convex structures for the convex hull $C^{2p}$
of trigonometric moment curve,
for arbitrary $p=3,4,\dots$. We show that $C^{2p}$ arises as a face
of the convex set $\mathbb S_{2\times p}$, and has the similar property as $C^4$.
It is possible to express any point in $C^{2p}$
by the convex combination of $k$ extreme points with $k\le p+1$, and
it is an interior point of $C^{2p}$ if and only if we need $p+1$ extreme points for the expression.
We can also characterize elements of the convex body $C^{2p}$ in $\mathbb C^p$
by a finite set of inequalities involving $p$ complex variables.
The convex hull of the symmetric trigonometric moment curve has been studied in \cite{{bar_2008},{vinzant}},
where only odd multiples of $t$ arise in (\ref{tri-moment}).

After we briefly review the classification of faces of the convex set $\mathbb S_{2\times 2}$ in the next section, we
consider the intersections of maximal faces in Section 3.
We examine $4$-dimensional faces of $\mathbb S_{2\times 2}$ in Section 3, to recover various moment
and trigonometric moment curves for the case of $p=2$.
In the final section, we construct the trigonometric moment curve for arbitrary $p=1,2,3,\dots$, which
generates a face of the convex set $\mathbb S_{2\times p}$. With this machinery, we characterize the whole facial structure
of the convex hull generated by the trigonometric moment curve.
Throughout this note, we will identify vectors of the vector space $\mathbb
C^2$ and points of $\mathbb C\mathbb P^1$. Two
points $x$ and $y$ of $\mathbb C\mathbb P^1$ coincide if and only if
they are parallel to each others, $x\parallel y$ in notation, as
vectors in $\mathbb C^2$.

The author is grateful to Lin Chen and Kil-Chan Ha for useful comments on the draft.

\section{Faces}\label{sec:sep}

For $\varrho=\varrho_1\ot\varrho_2\in M_m\ot M_n$, we define the partial transpose $\varrho^\Gamma$ by
$\varrho^\Gamma=\varrho_1^\ttt\ot\varrho_2$, where $\varrho_1^\ttt$ denotes the usual transpose of $\varrho_1$.
This operation can be extended on the whole $M_m\ot M_n$.
It was observed by Choi \cite{choi-ppt} and rediscovered by Peres \cite{peres} that if $\varrho$ is separable then $\varrho^\Gamma$
is still positive semi-definite. In other words, we have $\mathbb S_{m\times n}\subset \mathbb T_{m\times n}$, if
we denote by $\mathbb T_{m\times n}$ the set of all states $\varrho$ on $M_m\ot M_n$
with positive semi-definite partial transposes $\varrho^\Gamma$.
It was also shown in \cite{choi-ppt,horo-1,woronowicz} that $\mathbb S_{m\times n}= \mathbb T_{m\times n}$
if and only if $mn\le 6$. We exploit this result to investigate the structures of $\mathbb S_{2\times 2}=\mathbb T_{2\times 2}$.

It was shown in \cite{ha_kye_04} that every face of $\mathbb T_{m\times n}$ is determined by a pair $(D,E)$ of subspaces
in $\mathbb C^m\ot\mathbb C^n$, and is of the form
$$
\tau(D,E)=\{\varrho\in\mathbb T_{m\times n}: {\mathcal R}\varrho\subset D,\ {\mathcal R}\varrho^\Gamma\subset E\},
$$
where ${\mathcal R}\varrho$ denotes the range space of $\varrho$.
The pair $(D,E)$ is uniquely determined if we require
that $D$ and $E$ are minimal. It is very difficult in general to characterize
pairs $(D,E)$ for which $\tau(D,E)$ is nonempty. In the case of
$\mathbb S_{2\times 2}$, we have the complete list \cite{ha_kye_04}
of such pairs. We will identify a vector $(a,b,c,d)\in\mathbb C^2\ot\mathbb C^2$ with the $2\times 2$ matrix by
\begin{equation}\label{iissoo}
(a,b,c,d)\ \leftrightarrow\ \left(\begin{matrix}a&b\\c&d\end{matrix}\right).
\end{equation}
In this way, we can say about the rank of a vector in $\mathbb C^2\ot\mathbb C^2$.
A vector in $\mathbb C^2\ot\mathbb C^2$ is a product vector if and only if it is of rank one.
In the following, we list up all faces of the convex set $\mathbb S_{2\times 2}$, where
two natural numbers in the subscripts denote $\dim D$ and $\dim E$ for $\tau(D,E)$, respectively.
\begin{itemize}
\item
$G_{4,4}$: the full convex set $\mathbb S_{2\times 2}$.
\item
$G_{3,4}(V)=\tau(V^\perp, \{0\}^\perp)$, with a rank two vector $V\in\mathbb C^2\ot\mathbb C^2$.
\item
$G_{4,3}(W)=\tau(\{0\}^\perp, W^\perp)$, with a rank two vector $W\in\mathbb C^2\ot\mathbb C^2$.
\item
$H_{3,3}(x,y)=\tau((x\ot y)^\perp,(\bar x\ot y)^\perp)$, with $x,y\in\mathbb C^2$.
\item
$G_{3,3}(V,W)=\tau(V^\perp, W^\perp)$, with rank two vectors $V,W\in\mathbb C^2\ot\mathbb C^2$ with a suitable relation,
which will be given in the next section.
\item
$H_{2,2}(x,y,z,w)=\tau(\spa\{x\ot y,z\ot w\},\spa\{\bar x\ot y,\bar z\ot w\})$, with $x,y,z,w\in\mathbb C^2$ satisfying
 $x\parallel z$ or $y\parallel w$.
\item
$G_{2,2}(x,y,z,w)=\tau(\spa\{x\ot y,z\ot w\},\spa\{\bar x\ot y,\bar z\ot w\})$, with $x,y,z,w\in\mathbb C^2$ satisfying
 $x\nparallel z$ and $y\nparallel w$.
\item
$G_{1,1}(x,y)=\tau(\spa\{x\ot y\},\spa\{\bar x\ot y\})$ with $x,y\in\mathbb C^2$.
\end{itemize}

We note that $G_{1,1}(x,y)$ is an extreme point of $\mathbb S_{2\times 2}$ which is nothing but the rank one projection $P_{x\ot y}$
onto a unit product vector $x\ot y$. The face $G_{2,2}(x,y,z,w)$ is an edge which is the convex hull of
two extreme points $G_{1,1}(x,y)$ and $G_{1,1}(z,w)$.
On the other hand, $G_{3,4}(V)$, $G_{4,3}(W)$ and $H_{3,3}(x,y)$ are maximal faces.

For $V\in\mathbb C^2\ot\mathbb C^2$, we note that $P_{x\ot y}\in G_{3,4}(V)$ if and only if
$x\ot y\perp V$ if and only if $x\perp V\bar y$. Here, we consider $V$ as a $2\times 2$ matrix by (\ref{iissoo}),
and $\bar y$ denotes the vector whose entries are given by the complex conjugates of the corresponding entries.
We also note that
$x\perp V\bar y$ if and only if $(VW^{-1})^*x\perp W\bar y$. Therefore, the affine isomorphism
$$
\varrho\mapsto ((VW^{-1})^*\ot I_2)\varrho ((VW^{-1})\ot I_2)
$$
sends the face $G_{3,4}(V)$ onto the face $G_{3,4}(W)$, where $I_2=(1,0,0,1)\in\mathbb C^2\ot\mathbb C^2$.
The operation of partial transpose sends $G_{3,4}(V)$ onto $G_{4,3}(V)$. Therefore, all the faces of types $G_{3,4}$ and $G_{4,3}$ are
affinely isomorphic to each others. For given $x,y,z,w\in\mathbb C^2$, take nonsingular matrices $V$ and $W$ with
$Vz=x$ and $Ww=y$. Then we see that the affine isomorphism
$$
\varrho\mapsto (V^*\ot W^*)\varrho(V\ot W)
$$
maps the face $H_{3,3}(x,y)$ onto the face $H_{3,3}(z,w)$.

Now, we find extreme points for a given maximal face. We first note that all extreme points are parameterized by the
$4$-dimensional manifold $\mathbb C\mathbb P^1\times\mathbb C\mathbb P^1$. We see that an extreme point $G_{1,1}(z,w)$ belongs
to the face $H_{3,3}(x,y)$ if and only if $x\perp z$ or $y\perp w$. Therefore, extreme points of the face $H_{3,3}(x,y)$
are parameterized by the union of two $\mathbb C\mathbb P^1$'s in $\mathbb C\mathbb P^1\times\mathbb C\mathbb P^1$.
On the other hand, we see that extreme points of the face $G_{3,4}(V)$ are parameterized by $\mathbb C\mathbb P^1$. Indeed,
for each $z\in\mathbb C\mathbb P^1$, there is a unique
$w\in\mathbb C\mathbb P^1$ such that $z\perp V\bar w$, and so $P_{z\ot w}\in G_{3,4}(V)$.

\newcommand\cii{\circle*{0.3}}
\begin{center}
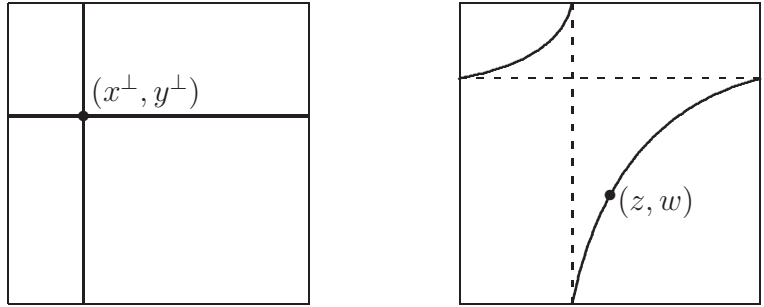
\begin{figure}[h]\label{class}
\setlength{\unitlength}{.5 truecm}
\begin{picture}(20,9)
\put(0,0){\line(1,0){8}} \put(0,0){\line(0,1){8}}
\put(0,8){\line(1,0){8}} \put(8,0){\line(0,1){8}}
\put(2.2,5.4){$(x^\perp, y^\perp)$}
\put(2,5){\cii}

\put(16.2,2.5){$(z,w)$}
\put(16,2.9){\cii}

\put(12,0){\line(1,0){8}} \put(12,0){\line(0,1){8}}
\put(12,8){\line(1,0){8}} \put(20,0){\line(0,1){8}}
\dashline{0.2}(12,6)(20,6)
\dashline{0.2}(15,0)(15,8)

\thicklines
\put(0,5){\line(1,0){8}} \put(2,0){\line(0,1){8}}
\qbezier(12,6)(14.7,6.5)(15,8)
\qbezier(15,0)(16,5)(20,6)
\end{picture}
\caption{
Points of $\mathbb C\mathbb P^1\times \mathbb C\mathbb P^1$ whose associated extreme points belong to
the faces $H_{3,3}(x,y)$ and $G_{3,4}(V)$. Both horizontal and vertical axes represent $\mathbb C\mathbb P^1$.}
\end{figure}
\end{center}

\section{intersection of maximal faces}\label{sec:intersection}

In this section, we consider intersections of maximal faces. We recall \cite{ha_kye_04} that every face of $\mathbb T_{m,n}$ is exposed,
and so it is the intersection of a family of maximal faces.
It is clear that the intersection of any two maximal faces of the form $H_{3,3}$ has two or infinitely many extreme points.
The intersection $H_{3,3}(x,y)\cap H_{3,3}(z,w)$ has infinitely many extreme points if and only if
$x=z$ or $y=w$ as elements of $\mathbb C\mathbb P^1$.
In this case, we have $H_{3,3}(x,y)\cap H_{3,3}(z,w)=H_{2,2}(x,y,z,w)$. This is affinely isomorphic to the convex set $\mathbb S_{1,2}$
which is realized as the $3$-dimensional solid sphere whose extreme points are parameterized by $\mathbb C\mathbb P^1$.
If $x\neq z$ or $y\neq w$ then the intersection is an edge.

It is also apparent that the intersection $H_{3,3}(x,y)\cap G_{3,4}(V)$ has one or two extreme points. It has a single extreme point
if and only if $(x^\perp, y^\perp)$ is an extreme point of $G_{3,4}(V)$. The intersection $H_{3,3}(x,y)\cap G_{4,3}(W)$ is similar.
As for the intersection $G_{3,4}(V)\cap G_{3,4}(W)$, we have the following:

\begin{proposition}
Let $V$ and $W$ be non-singular matrices which are different up to scalar multiplication.
Then $G_{3,4}(V)\cap G_{3,4}(W)$ is a single point or a line segment. It is a single point
if and only if $W^{-1}V=(a,b,c,d)$ satisfies the following two conditions:
\begin{enumerate}
\item[(i)]
$b\neq 0$ or $c\neq 0$
\item[(ii)]
$(a-d)^2+4bc=0$.
\end{enumerate}
\end{proposition}

\begin{proof}
An extreme point $P_{x\ot y}$ belongs to
$G_{3,4}(V)\cap G_{3,4}(W)$ if and only if $x\perp V\bar y$ and $x\perp W\bar y$ if and only if
\begin{equation}\label{vvv}
W^{-1}V\bar y\parallel \bar y.
\end{equation}
We first consider the case $b=c=0$. In this case, we have $a\neq d$, and $\bar y=(1,0)$ and $\bar y=(0,1)$ are
always two solutions of (\ref{vvv}). Hence, the assertion is true.
If $b=0$ and $c\neq 0$ then $\bar y=(0,1)$ is always a solution of (\ref{vvv}). Furthermore, (\ref{vvv}) has a solution of the form
$\bar y=(1,\xi)$ if and only if such solution is unique if and only if $a-d\neq 0$. Therefore, we get the required conclusion.
The same reasoning may be applied for the case $b\neq 0$ and $c=0$. Finally, suppose that $bc\neq 0$. In this case, any solution
of (\ref{vvv}) must be of the form $\bar y=(1,\xi)$, and this is a solution if and only if $b\xi^2+(a-d)\xi-c=0$. From
this, we get the conclusion.
\end{proof}

We have the similar result for the intersection $G_{4,3}(V)\cap G_{4,3}(W)$, and it remains to consider the intersection
$G_{3,4}(V)\cap G_{4,3}(W)$.
We first note that the extreme point $P_{x\ot y}$ belongs to the intersection if and only if
\begin{equation}\label{ffff}
x\perp V\bar y,\qquad \bar x\perp W\bar y,
\end{equation}
because $(P_{x\ot y})^\Gamma=P_{\bar x\ot y}$. This is equivalent to $V\bar y\parallel \bar W y$ or
$\bar W^{-1}V\bar y\parallel y$. The matrix of the form in the statement (v) in the following lemma appears in
Proposition 3.6 of \cite{byeon-kye}, to classify faces of the convex cone of all positive linear maps from $M_2$ into $M_2$.

\begin{lemma}\label{lemma}
Let $A$ be a $2\times 2$ nonsingular matrix. Then the possible numbers of the solutions of
\begin{equation}\label{eqzz}
A\bar y\parallel y
\end{equation}
in $\mathbb C\mathbb P^1$ are $0,1,2,\infty$, and the following are equivalent:
\begin{enumerate}
\item[(i)]
The equation {\rm (\ref{eqzz})} has infinitely many solutions in $\mathbb C\mathbb P^1$.
\item[(ii)]
The set of solutions of {\rm (\ref{eqzz})} in $\mathbb C\mathbb P^1$ is a circle.
\item[(iii)]
$A$ is of the form $BJ\bar B^{-1}$ for a nonsingular $B$, where $J=\left(\begin{matrix}0&1\\1&0\end{matrix}\right)$.
\item[(iv)]
$A$ is of the form $C\bar C^{-1}$ for a nonsingular $C$.
\item[(v)]
$A$ is of the form $\left(\begin{matrix}\alpha &r\\s&-\bar\alpha\end{matrix}\right)$, where $\alpha\in\mathbb C$
and $r,s\in\mathbb R$  with $|\alpha|^2+rs>0$.
\end{enumerate}
Furthermore, for a given circle on the sphere $\mathbb C\mathbb P^1$, there is a unique nonsingular matrix $A$
up to scalar multiplications such that
the solution of the equation {\rm (\ref{eqzz})} consists of the circle.
\end{lemma}

\begin{proof}
We denote by $a_{ij}$ the $ij$-entry of $A$. We first note that the equation (\ref{eqzz}) with $y=(\xi,1)^\ttt$ becomes
\begin{equation}\label{eqxxx}
a_{21}|\xi|^2+a_{22}\xi-a_{11}\bar\xi-a_{12}=0.
\end{equation}
We also note that $y=(1,0)^\ttt$ is a solution of (\ref{eqzz}) if and only if $a_{21}=0$. In this case, both real and imaginary parts
of (\ref{eqzz}) with $y=(\xi,1)^\ttt$ represent lines on the complex planes. If $a_{21}\neq 0$ then
both real and imaginary parts
of (\ref{eqzz}) with $y=(\xi,1)^\ttt$ represent lines or circles. In any cases, the solution of the equation
(\ref{eqzz}) in the sphere $\mathbb C\mathbb P^1$ is the intersection of two circles on the sphere. This tells us that
the possible numbers of the solutions in $\mathbb C\mathbb P^1$ are $0,1,2,\infty$, and shows the direction (i) $\Longrightarrow$ (ii).

For the direction (ii) $\Longrightarrow$ (iii), suppose that the solution of (\ref{eqzz}) is the circle on the sphere
$\mathbb C\mathbb P^1$, which is represented by the M\"obius transform
$$
\omega\mapsto \left(\dfrac{a\omega+b}{c\omega+d},\ 1\right)^\ttt
=\left( a\omega+b,c\omega+d\right)^\ttt\in\mathbb C\mathbb P^1,\qquad |\omega|=1,
$$
where $B=\left(\begin{matrix} a&b\\c&d\end{matrix}\right)$ is nonsingular. This means that
$A\bar B(\bar \omega, 1)^\ttt\parallel B(\omega,1)^\ttt$,
or equivalently, $B^{-1}A\bar B(\bar \omega, 1)^\ttt\parallel (\omega,1)^\ttt$
holds for each $\omega$ with $|\omega|=1$. From this, it is easy to see that
$B^{-1}A\bar B=J$, up to scalar multiplications.

The direction (iii) $\Longrightarrow$ (v) is a direct calculation. Indeed, we have
$$
A=BJ\bar B^{-1}=(\det \bar B)^{-1}\left(\begin{matrix}
-a\bar c+b\bar d & |a|^2-|b|^2\\
-|c|^2+|d|^2&\bar ac-\bar bd\end{matrix}\right).
$$
Put $\alpha=-a\bar c+b\bar d$, $r=|a|^2-|b|^2$ and $s=-|c|^2+|d|^2$. Then we see that $|\alpha|^2+rs=|ad-bc|^2>0$,
since $B$ is nonsingular.

For the direction (ii) $\Longrightarrow$ (iv), we represent the circle by the image of the M\"obius transform
$$
t\mapsto\left(\dfrac{at+b}{ct+d},\ 1\right)^\ttt
=\left( at+b,ct+d\right)^\ttt\in\mathbb C\mathbb P^1,\qquad t\in\mathbb R
$$
on the real line, with a nonsingular $C=\left(\begin{matrix}
a&b\\c&d\end{matrix}\right)$. The exactly same argument as above
shows that $C^{-1}A\bar C=I_2$, up to scalar multiplications. The
direction (iv) $\Longrightarrow$ (v) is also same as (iii)
$\Longrightarrow$ (v)

It remains to show the direction (v) $\Longrightarrow$ (i).
Suppose that $A$ is of the form in (v). Then we see that $A(\bar\xi,1)^\ttt\parallel (\xi,1)$ if and only if
$$
s|\xi|^2-\bar\alpha\xi-\alpha-\bar\xi-r=0.
$$
If $s\neq 0$, then this is the circle
$\left|\xi-\dfrac {\alpha}s\right|^2=\dfrac {|\alpha|^2+rs}{s^2}$. If $s=0$, then this the line $2{\text{\rm Re}}(\bar\alpha\xi)=s$.
Therefore, {\rm (\ref{eqzz})} has infinitely many solutions.

For the last claim, we note that three mutually distinct points $(\xi_i,1)^\ttt$ on the sphere, $i=1,2,3$, determine a circle.
We consider the vectors
$$
\Xi_i=(\xi_i,-1)\ot (\bar\xi_i,1)=(|\xi_i|^2,\xi_i,-\bar\xi_i,-1)\in\mathbb C^4.
$$
By Proposition 2.1 of \cite{ha+kye_unique_decom}, we see that $\{\Xi_1,\Xi_2,\Xi_3\}$ is linearly independent. Therefore,
the matrix $A$ is determined by (\ref{eqxxx}) up to scalar multiplications.
\end{proof}

If $r=s=0$ then we have
$$
e^{-i\theta}\left(\begin{matrix}e^{i\theta}&0\\0&-e^{-i\theta}\end{matrix}\right)
=\left(\begin{matrix}1&0\\0&-e^{-2i\theta}\end{matrix}\right).
$$
Therefore, any diagonal matrix satisfies the conditions of Lemma \ref{lemma} by (v),
whenever two diagonals have the same absolute values.

The relation between the matrices $B$ and $C$ in Lemma \ref{lemma} is now clear. We denote by $\phi_B$ the M\"obius transformation
given by $B$. Then we see that the following are equivalent:
\begin{enumerate}
\item[(i)]
The image of the unit circle under $\phi_B$ coincides with the image of the real axis under $\phi_C$.
\item[(ii)]
$\phi_{C^{-1}B}=\phi_C^{-1}\circ\phi_B$ sends the unit circle onto the real axis.
\item[(iii)]
$BJ\bar B^{-1}=C\bar C^{-1}$, or equivalently $B^{-1}C\, \overline{B^{-1}C}^{-1}=J$, up to scalar multiplications.
\end{enumerate}
As a byproduct, we see that $\phi_D$ sends the unit circle onto the real axis if and only if $D^{-1}\bar D=J$ up to
scalar multiplications.

\begin{theorem}
Let $(V,W)$ be a pair of $2\times 2$ nonsingular matrices which are not parallel to each other.
Then $G_{3,4}(V)\cap G_{4,3}(W)=G_{3,3}(V,W)$ if and only if $\bar W^{-1}V$ is of the form in Lemma \ref{lemma}.
If $(V_1, W_1)$ and $(V_2,W_2)$ are pairs satisfying these conditions, then $G_{3,3}(V_1,W_1)=G_{3,3}(V_2,W_2)$
if and only if $\bar W_1^{-1}V_1=\bar W_2^{-1}V_2$. Especially, we have
$G_{3,3}(V,W)=G_{3,3}(\bar W^{-1}V,I)$.
\end{theorem}

\section{Moment curves arising from two qubit system}\label{sec:moment-curve}

In this section, we explain how the moment curve arises. To do this, we need to consider the convex cone $\mathbb S^\prime_{2\times 2}$
generated by the convex set $\mathbb S_{2\times 2}$. This is the convex hull of all positive semi-definite
rank one matrices onto a product vector which is not necessarily normalized. We see that
$$
\mathbb S_{2\times 2}=\{\varrho\in \mathbb S^\prime_{2\times 2}: \tr\varrho=1\}.
$$
Faces of the convex cone $\mathbb S^\prime_{2\times 2}$ correspond to faces of $\mathbb S_{2\times 2}$ in an obvious way. We denote
by $G^\prime_{k,\ell}(V,W)$ the face of $\mathbb S^\prime_{2\times 2}$ which corresponds to the face $G_{k,\ell}(V,W)$ of $\mathbb S_{2\times 2}$.

Now, we suppose that $G_{3,4}(V)\cap G_{4,3}(W)=G_{3,3}(V,W)$.
The extreme points in this convex set can be described in two ways, using the unit circle or the real axis,
as it was shown in the proof of Lemma \ref{lemma}.
We first consider the unit circle. We note that $P_{x\ot y}$ is an extreme point if and only if $y$ is of the form
$$
y_\omega=\left( a\omega+b,c\omega+d\right)^\ttt\in\mathbb C\mathbb P^1,
$$
for a complex number $\omega$ of modulus one.
From the relation $x_\omega\perp \bar W y_\omega$, we also see that entries of $x_\omega$ must be linear combinations of $\bar\omega$ and $1$.
Therefore, all entries of $x_\omega\ot y_\omega$ are linear combinations of $\omega, \bar\omega$ and $1$, and all
entries of the $4\times 4$ matrix representing the rank one positive semi-definite matrix
$$
R_{x_\omega\ot y_\omega}=(x_\omega\ot y_\omega)(x_\omega\ot y_\omega)^*\in G_{3,3}^\prime(V,W)
$$
are also linear combinations of
$\omega,\bar\omega,\omega^2,\bar\omega^2$ and $1$. Therefore, it is now clear that
$$
(\omega,\omega^2)\mapsto R_{x_\omega\ot y_\omega}
$$
extends to an affine isomorphism from the convex body $C^4\subset\mathbb R^4$ into $G^\prime_{3,3}(V,W)\subset\mathbb S_{2\times 2}$,
which is the convex cone generated by
the image of this isomorphism.
The map
$$
\omega\mapsto (\omega,\omega^2)
$$
is nothing but the trigonometric moment curve with $p=2$, as it was explained in Introduction.

Now, we consider the face $G_{3,3}(V,W)$ with the specific example
\begin{equation}\label{VW}
V=(1,0,0,-1)=
\left(\begin{matrix}1&0\\0&-1\end{matrix}\right),
\qquad
W=(0,1,-1,0)=
\left(\begin{matrix}0&1\\-1&0\end{matrix}\right),
\end{equation}
for which there are infinitely many solutions for (\ref{ffff}), since $\bar W^{-1}V=J$.
We note that every $4\times 4$ Hermitian matrix $\varrho$ with a kernel vector $V$ must be of the form
\begin{equation}\label{xxx}
\varrho
=\left(\begin{matrix}
a&\bar\delta&\alpha&a\\
\delta&b&\gamma&\delta\\
\bar\alpha&\bar\gamma&c&\bar\alpha\\
a&\bar\delta&\alpha&a
\end{matrix}\right),
\qquad {\text{\rm with}}\
\varrho^\Gamma
=\left(\begin{matrix}
a&\bar\delta&\bar\alpha&\bar\gamma\\
\delta&b&a&\bar\delta\\
\alpha&a&c&\bar\alpha\\
\gamma&\delta&\alpha&a
\end{matrix}\right).
\end{equation}
Since $W$ is a kernel vector of $\varrho^\Gamma$, we see that $a=b=c=\dfrac 14$ and $\delta=\alpha$. Therefore, $\varrho$
and $\varrho^\Gamma$ are of the form
\begin{equation}\label{4x4}
\varrho(\alpha,\gamma)
=\dfrac 14\left(\begin{matrix}
1&\bar\alpha&\alpha&1\\
\alpha&1&\gamma&\alpha\\
\bar\alpha&\bar\gamma&1&\bar\alpha\\
1&\bar\alpha&\alpha&1
\end{matrix}\right),
\qquad
\varrho(\alpha,\gamma)^\Gamma
=\dfrac 14\left(\begin{matrix}
1&\bar\alpha&\bar\alpha&\bar\gamma\\
\alpha&1&1&\bar\alpha\\
\alpha&1&1&\bar\alpha\\
\gamma&\alpha&\alpha&1
\end{matrix}\right).
\end{equation}
Considering the determinants of principal submatrices, we
finally see that both $\varrho$ and $\varrho^\Gamma$ are positive semi-definite if and only if
\begin{equation}\label{C^4}
|\alpha|\le 1,\qquad |\gamma|\le 1,\qquad
2|\alpha|^2+|\gamma|^2-\alpha^2\bar\gamma-\bar\alpha^2\gamma\le 1.
\end{equation}
Therefore, the face $G_{3,3}(V,W)$ is determined by $4\times 4$ matrices in (\ref{4x4}) whose entries satisfy
the three inequalities in (\ref{C^4}).

Now, we see that an extreme point $P_{x\ot y}$ belongs to $G_{3,3}(V,W)$ if and only if
$$
x=(1,\bar\omega),\qquad y=(1,\omega)
$$
for a complex number $\omega$ with $|\omega|=1$. These extreme points are realized by the following $4\times 4$ matrix
$$
P_{x\ot y}=
\dfrac 14\left(\begin{matrix}
1&\bar\omega&\omega&1\\
\omega&1&\omega^2&\omega\\
\bar\omega&\bar\omega^2&1&\bar\omega\\
1&\bar\omega&\omega&1
\end{matrix}\right)=\varrho(\omega,\omega^2),
$$
with $|\omega|=1$. We note that the map
\begin{equation}\label{iso_C^4}
(\alpha,\gamma)\mapsto
\varrho(\alpha,\gamma),
\end{equation}
defines an affine isomorphism from $C^4\subset\mathbb R^4$ onto $G_{3,3}(V,W)\subset\mathbb T_{2\times 2}$,
which sends extreme points $(\omega,\omega^2)$ of $C^4$
onto extreme points $P_{x_\omega\ot y_\omega}=\varrho(\omega,\omega^2)$ of $G_{3,3}(V,W)$.
Finally, we note that $(\alpha,\gamma)\in C^4$ if and only if
$\varrho(\alpha,\gamma)\in G_{3,3}(V,W)$ if and only if $(\alpha,\gamma)$ satisfies the relation (\ref{C^4}).
This gives us the description of the convex body $C^4$ by inequalities.

If the circle is represented by $y_t=(at+b,ct+d)^\ttt$ with $t\in\mathbb R$, then
all entries of $R_{x_t\ot y_t}$ are linear combinations of $1,t,t^2,t^3$ and $t^4$. Therefore, the map
$$
(t,t^2,t^3,t^4)\mapsto R_{x_t\ot y_t}
$$
extends to an affine isomorphism from the convex hull of the moment curve into a face of the form $G^\prime_{3,3}$, which is
the convex cone generated by the image. For a concrete example, we take $C=I$ in Lemma \ref{lemma}, and consider
the face $G^\prime_{3,3}(I,I)$. Extreme points are given by
$$
R_{(1,-t)\ot (t,1)}=R_{(t,1,-t^2,-t)}
=\left(\begin{matrix}
t^2& t& t^3 & -t^2\\
t & 1 & t^2 & -t\\
t^3& t& t^4 & -t^3\\
-t^2& -t& -t^3& t^2
\end{matrix}\right),
\qquad t\in\mathbb R.
$$
In this way, we have a concrete realization of the convex body generated by the moment curve.

Now, we pay attention to the face $G_{3,4}(V)$ itself, with $V$ given
by (\ref{VW}). We see that an extreme point $P_{x\ot y}$ belongs to the face $G_{3,4}(V)$ if and only if
$x\perp V\bar y$ if and only if $x\ot y$ is of the form $(x_1,x_2)\ot (x_2,x_1)$. Therefore, extreme points of
the face $G_{3,4}(V)$ is determined by $x=(x_1,x_2)\in\mathbb C^2$, and the corresponding rank one matrix
$\varrho_x=R_{(x_1,x_2)\ot (x_2,x_1)}$ is of the form
$$
\varrho_x=
\left(
\begin{matrix}
|x_1x_2|^2  & |x_1|^2\bar x_1 x_2 & |x_2|^2 x_1\bar x_2  & |x_1x_2|^2\\
|x_1|^2 x_1\bar x_2  & |x_1|^4   & x_1^1\bar x_2^2  &|x_1|^2 x_1\bar x_2\\
|x_2|^2 \bar x_1 x_2 &\bar x_1^1 x_2^2  &|x_2|^4 &  |x_2|^2 \bar x_1 x_2\\
|x_1x_2|^2 &|x_1|^2 \bar x_1 x_2 & |x_2|^2  x_1 \bar x_2  & |x_1x_2|^2
\end{matrix}
\right)
$$
Note that this is of the form in (\ref{xxx}) as it is expected. If $\|x\|=1$ then $\varrho_x$ is of trace one. Furthermore,
if two unit vectors $x$ and $z$ are parallel to each other then $\varrho_x=\varrho_z$. Therefore, the map
$$
x\mapsto \varrho_x
$$
gives rise to a homeomorphism from $\mathbb C\mathbb P^1$ into $G_{3,4}(V)$. If we write
$$
S(x)=(|x_1|^4,|x_2|^4, |x_1|^2x_1\bar x_2, |x_2|^2\bar x_1 x_2, x_1^2\bar x_2^2)\in\mathbb R^2\times\mathbb C^3=\mathbb R^8,
$$
then we see that $S(x)\leftrightarrow \varrho_x$ is an affine isomorphism. If we parameterize $\mathbb C\mathbb P^1$ by
the spherical coordinate $x= (\cos\phi,\sin\phi\cos\theta,\sin\phi\sin\theta)$
by $x_1=\cos\phi$ and $x_2=e^{-i\theta}\sin\phi$, then we have
the {\sl moment surface} given by
$$
\begin{aligned}
S(x)&=
(\cos^4\phi,\sin^4\phi,\cos^3\phi\sin\phi\cos\theta,\cos^3\phi\sin\phi\sin\theta,\\
&\phantom{X}
\cos\phi\sin^2\phi\cos\theta,-\cos\phi\sin^2\phi\sin\theta,
\cos^2\phi\sin^2\phi\cos 2\theta,\cos^2\phi\sin^2\phi\sin 2\theta)\in\mathbb R^8.
\end{aligned}
$$
The convex body $S^8$ in $\mathbb R^8$ generated by the image of $S$ is affinely isomorphic
to the maximal face $G_{3,4}(V)$.
A point in the convex body $S_8$ is an extreme point if and only if it is a point on the surface.
The convex combination of any two extreme points is an edge. The image of circles in $\mathbb C\mathbb P^1$ under the map $S$
generate maximal faces of $S_8$, which are of $4$-dimensional. There is no more nontrivial face of $S_8$.

\section{Moment curves arising from qubit-qudit system}\label{sec:moment-curve-2xp}

In this final section, we construct the trigonometric curve
(\ref{tri-moment}) from a face of the convex set $\mathbb S_{2\times
p}$ of all $2\otimes p$ separable states. To do this, we define
$2\times p$ matrices $V_i$ and $W_i$  by
\begin{equation}\label{vw}
V_i=e_{1,i}-e_{2,i+1}\in M_{2\times p},\qquad
W_i=e_{1,i+1}-e_{2,i}\in M_{2\times p},\qquad i=1,2,\dots,p-1.
\end{equation}
where $\{e_{i,j}\}$ denotes the standard matrix units of $M_{2\times p}$.
Suppose that $x\otimes y\in\mathbb C^2\ot\mathbb C^p$ belongs to the face $\tau(\{V_i\}^\perp, \{W_i\}^\perp)$
of $\mathbb T_{2\times p}$ then we have
\begin{equation}\label{bbb}
V_i^*x\perp \bar y,\qquad  x\perp \bar W_i y.
\end{equation}
Note that $V_i^*x=x_1e_i-x_2e_{i+1}\in\mathbb C^p$ for $i=1,2,\dots,p-1$, with the standard orthonormal basis
$\{e_i\}$ of $\mathbb C^p$. We see that $\{V_i^*x:i=1,\dots,p-1\}$ is linearly independent
for every nonzero $x\in\mathbb C^2$, and so $y$ is uniquely determined
by $x$, up to scalar multiplications. We actually have
$$
y= x_2^{p-1}e_1+ x_1x_2^{p-2}e_2+  \dots x_1^{p-2}x_2e_{p-1} + x_1^{p-1} e_p\in\mathbb C^p.
$$
Since $\bar W_i y=(x_1^i x_2^{p-(i+1)}, -x_1^{i-1}x_2^{p-i})^\ttt\in\mathbb C^2$,
we see that $x\perp \bar W_i y$ holds if and only if
$$
|x_1|^2x_1^{i-1}x_2^{p-(i+1)}
=
|x_2|^2x_1^{i-1}x_2^{p-(i+1)},
\qquad i=1,2,\dots,p-1.
$$
From this, we see that the solution of the equation (\ref{bbb}) is given by
$$
x_\omega=(1,\bar\omega)^\ttt\in\mathbb C\mathbb P^1,\qquad
y_\omega=(1,\omega,\omega^2,\dots\omega^{p-1})^\ttt\in\mathbb C\mathbb P^{p-1},\qquad |\omega|=1,
$$
and the corresponding pure product states are of the following form
$$
P_\omega:=\frac 1{2p}P_{x_\omega\ot y_\omega}
=\frac 1{2p}\left(\begin{matrix}A&B\\B^*&A\end{matrix}\right)\in \mathbb S_{2\times p}\subset M_2\otimes M_p,
$$
with
$$
A=\left(\begin{matrix}
1 & \bar\omega & \cdots & \bar\omega^{p-2}& \bar\omega^{p-1}\\
\omega &  &  & & \bar\omega^{p-2}\\
\vdots &  & \ddots &&\vdots\\
\omega^{p-2}& & & & \bar\omega\\
\omega^{p-1}&\omega^{p-2} &\cdots & \bar\omega& 1\end{matrix}\right),
\qquad
B=\left(\begin{matrix}
\omega & 1 & \cdots &\bar\omega^{p-3}&\bar\omega^{p-2}\\
\omega^2 & & &&\bar\omega^{p-3}\\
\vdots &  & \ddots &&\vdots\\
\omega^{p-1} &  & &&1\\
\omega^{p} &\omega^{p-1}  &\cdots &\omega^2&\omega\end{matrix}\right).
$$
Therefore, we see that
$$
(\omega,\omega^2,\dots,\omega^p)\longleftrightarrow P_{\omega}
$$
extends to an affine isomorphism from the convex set $C^{2p}$ onto the convex hull of pure product states $P_{\omega}$.
The homeomorphism
$$
\omega\mapsto (\omega,\omega^2,\dots,\omega^p)\in\mathbb C^p=\mathbb R^{2p},\qquad |\omega|=1
$$
is nothing but the trigonometric moment curve given by (\ref{tri-moment}).

\begin{theorem}\label{main-thzz}
The convex hull of pure product states $P_{\omega}$
coincides with the face $\tau(D,E)$ of the convex set $\mathbb
T_{2\times p}$.
\end{theorem}

\begin{proof}
We denote by $C$ the convex hull of $P_{\omega}$, then we see
that $C=\tau(D\cap E)\cap\mathbb S_{2\times p}$ is a subset of $\tau(D,E)$.
We also denote by $\{\zeta^i:i=0,1,\dots p\}$ the $(p+1)$-th root of unity.
For any $\omega$ with $|\omega|=1$, we have
$\sum_{i=0}^p(\omega\zeta^i)^k=\sum_{i=0}^p(\zeta^i)^k=0$ for $k=1,2,\dots,p$, and so it follows that
$$
\sum_{i=0}^pP_{\zeta^i}=\sum_{i=0}^p P_{\omega\zeta^i}
=P_\omega+\sum_{p=1}^p P_{\omega\zeta^i}.
$$
Since $P_\omega$ is an arbitrary extreme point of $C$, this means that
$$
\varrho_I:=\frac 1{p+1}\sum_{i=0}^pP_{\zeta^i}=\frac1{2p(p+1)}\left(\begin{matrix} I_p&J\\J^*& I_p\end{matrix}\right)\in M_2\ot M_p
$$
is an interior point of $C$, where $J=\sum_{i=1}^{p-1}e_{i,i+1}\in M_p$.
It is easy to check that the range spaces of $\varrho_I$ and $\varrho_I^\Gamma$ coincide with $D$ and $E$, respectively.
Therefore, $\varrho_I$ is a common interior point of both $C$ and $\tau(D,E)$, and so we have
$\inte C\subset \inte\tau(D,E)$.

Now, assume that $C$ is a proper subset of $\tau(D,E)$. Then there exists a maximal face $F$ of $C$
such that $\inte F\subset\inte\tau(D,E)$. Take $\varrho_1\in\inte F$ and an extreme point $\varrho_0$ of $C$ which
is not in $F$. Then, by \cite{kye-canad} Proposition 2.5, we see that
$$
\varrho_t:=(1-t)\varrho_0+ t \varrho_1
$$
is an interior point of $C$ for every $t$ with $0<t<1$, and so, it is an interior point of $\tau(D,E)$.
Since $\varrho_1\in\inte F$ is an interior point of $\tau(D,E)$ there exist $a>1$ such that $\varrho_a\in\tau(D,E)$,
which is an entangled state, because $\varrho_a\notin C$.
Let $\mu$ be the maximum so that $\varrho_\mu\in\tau(D,E)$.
If $1<a<\mu$ then $\varrho_a$ is a interior point of $\tau(D,E)$, and so $\rk\varrho_a=\rk\varrho_a^\Gamma=p+1$.
Because $\varrho_a$ is the convex combination of $\varrho_0$ and $\varrho_\mu$, and $\varrho_0$ is an extreme point,
we conclude that $\rk\varrho_\mu=\rk\varrho_\mu^\Gamma=p$ by Theorem 3 of \cite{chen_dj_2xd}.

Finally, it is easy to see that if $\varrho\in\tau(D,E)$ is supported on $\mathbb C^2\ot \mathbb C^{p-1}$
in the sense of \cite{2xn} then
$\rk\varrho\le p-1$ and $\rk\varrho^\Gamma\le p-1$. Therefore, we see that $\varrho_\mu$ must be supported on $\mathbb C^2\ot\mathbb C^p$,
and so it is separable by \cite{2xn}. This contradiction completes the proof.
\end{proof}

Now, the convex hull $C^{2p}$ of the trigonometric moment curve is affinely isomorphic to the convex hull
of the image of the curve
$$
\omega\mapsto P_{\omega}\in\mathbb S_{2\times p},\qquad |\omega|=1.
$$
Therefore, we identity the convex set $C^{2p}$ and the face $\tau(D,E)$ of $\mathbb T_{2\times p}$
by Theorem \ref{main-thzz}. We look for proper faces of
the convex set $C^{2p}$, or equivalently the convex set $\tau(D,E)$.
It must be of the form
$\tau(D_1,E_1)$ for subspaces $D_1$ and $E_1$ of $D$ and $E$,
respectively, at least one of which is proper. Consider a matrix in $D\ominus D_1$ or $E\ominus E_1$,
then we see that if an extreme point $P_{\omega}$ belongs
to $\tau(D_1,E_1)$ then $\omega$ satisfies an equation given by the matrix,
which turns out to be a polynomial of degree $p$. This means that
the cardinality of extreme points of any proper face of $C^{2p}$ is
less than or equal to $p$. Therefore, we conclude that any boundary point of $C^{2p}$ is the convex combination
of $k$ extreme points with $k\le p$.
For the converse, if we take extreme points whose cardinality is less than $p+1$, then
it is clear that the convex hull $\varrho$ of them in $\tau(D,E)$ has the rank less than $p+1$. This means that the range space
of $\varrho$ is a proper subspace of $D$, and so we see that $\varrho$ is a boundary point of $C^{2p}$.
We conclude that
$$
(\omega_1,\dots,\omega_p)\mapsto \conv \{P_{\omega_1},\dots, P_{\omega_p}\}
$$
is a one-to-one correspondence from the $p$-dimensional torus onto the lattice of all nontrivial faces of $C^{2p}$,
where the lattice structure on the $p$-dimensional torus is given by the set inclusion of the entries of
the ordered sets $(\omega_1,\dots,\omega_p)$ representing points of the $p$-dimensional torus.
We also see that every face is exposed by \cite{ha_kye_04}.

If we take distinct
$\omega_1,\dots,\omega_k$ with $k\le p$, then we see that
$\{x_{\omega_i}\}$ is mutually distinct and $\{y_{\omega_i}\}$ is
linearly independent. By the result in \cite{alfsen,kirk}, we
conclude that the convex hull of $\{P_{\omega_i}:i=1,2,\dots,k\}$
is a simplex. Therefore, we see that any proper face of the convex set $C^{2p}$
is a simplex with $k$ extreme points, where $k=1,2,\dots, p$.
If $\varrho$ is an interior point of $\tau(D,E)$ then $\rk\varrho=\rk\varrho^\Gamma=p+1$,
and so we apply Theorem 3 of \cite{chen_dj_2xd} to conclude that $\varrho$ is the convex combination
of $p+1$ extreme points.

If we take any $k$ distinct complex numbers $\omega_i$ of modulus one with $k\le 2p+1$, then the corresponding
extreme points $P_{\omega_i}$ are linearly independent by Proposition 2.2 of \cite{ha+kye_unique_decom}.
Therefore, their convex hull must be a simplex. The number $2p+1$ is the maximum number of extreme points of $C^{2p}$
whose convex hull is a simplex, because the affine dimension of $C^{2p}$ is $2p$.
We summarize as follows:

\begin{theorem}\label{main-2p}
The convex hull $C^{2p}$ of the trigonometric moment curve has the following properties for any $p=1,2,\dots$.
\begin{enumerate}
\item[(i)]
A point of $C^{2p}$ is an extreme point if and only if it is on the trigonometric moment curve.
\item[(ii)]
Any point of $C^{2p}$ is the convex combination of $k$ extreme points with $k\le p+1$.
\item[(iii)]
If we take $k$ distinct extreme points with $k\le 2p+1$, then their convex combination is a simplex.
It is a face of $C^{2p}$ if and only if $k\le p$. Every face is exposed, and there are no more nontrivial faces.
\end{enumerate}
\end{theorem}

Now, we can express elements of $C^{2p}$ with $2p\times 2p$ matrices as in (\ref{xxx}) for $p=2$. These matrices
have $p$ complex variables.
If we consider the determinants of principal submatrices, then we get the finite set of inequalities as in (\ref{C^4}) which determines
elements of $C^{2p}$.


We also note that
the convex combination $\varrho$ of $\ell$ distinct extreme points with $\ell\ge p+1$ and nonzero coefficients
is an interior point by Theorem \ref{main-2p}.
If we take an arbitrary extreme point $P_\omega$
of the convex set $S^{2p}$, then the line segment from $P_\omega$
to $\varrho$ can be extended until it meets a boundary point which should be expressed uniquely as the convex combination
of $k$ extreme points with $k\le p$.  Therefore, the interior point $\varrho$ can be expressed uniquely as the convex combination
of $P_\omega$ and other $k$ extreme points with $k\le p$. The situation is visible
when $p=1$, with the picture of a circle on the plane.

Finally, we remark that it is also possible to construct trigonometric moment curve in the convex set $\mathbb S_{m\times n}$. To do this, we
begin with
$$
x_\omega=(1,\bar\omega,\bar\omega^2,\dots,\bar\omega^{m-1})^\ttt\in\mathbb C\mathbb P^{m-1},\qquad
y_\omega=(1,\omega,\omega^2,\dots,\omega^{n-1})^\ttt\in\mathbb C\mathbb P^{n-1}
$$
for a complex number $\omega$ of modulus one,
and consider the projection $P_\omega:=P_{x_\omega\ot y_\omega}$ onto the product vector $x_\omega\ot y_\omega$.
Then we see that the correspondence
$$
(\omega,\omega^2,\dots,\omega^{m+n-2})\longleftrightarrow P_\omega
$$
gives rise to an affine isomorphisms from the convex hull $C^{2(m+n-2)}$ of the trigonometric curve
onto the convex hull of $\{P_{\omega} :|\omega|=1\}$. For any interior point $\varrho$ of this convex hull, we have
$\rk\varrho=\rk\varrho^\Gamma=m+n-1$, as for the case of $C^{2p}$ arising from $\mathbb S_{2\times p}$.

\end{document}